\documentclass[11pt]{article}

\usepackage{amsmath}
\usepackage{amssymb}
\usepackage{enumerate}

\usepackage{amsthm}
\usepackage[margin=1in]{geometry}

\newtheorem{fact}{Fact}
\newtheorem{theorem}{Theorem}

\theoremstyle{definition}

\title{Probabilistic verifiers for asymmetric debates}

\author{H. G{\"{o}}kalp Demirci$ ^{1} $, A. C. Cem Say$ ^{1} $ and Abuzer Yakary{\i}lmaz$ ^{2,} $\footnote{Abuzer Yakary{\i}lmaz was partially supported by FP7 FET-Open project QCS.}
\\ \\
\small $ ^{1} $Bo\u{g}azi\c{c}i University, Department of Computer Engineering, Bebek 34342 \.{I}stanbul, Turkey
\\
\small $ ^{2} $University of Latvia, Faculty of Computing,  Raina bulv. 19, R\={\i}ga, LV-1586 Latvia
\\
\small \texttt{\{gokalp.demirci,say\}@boun.edu.tr,abuzer@lu.lv}
\\ \\
\today
}

\date{\small Keywords: 
partial information debates, 
games, 
alternation, 
probabilistic computation}

\begin{document}

\maketitle

\begin{abstract}
We examine the power of silent constant-space probabilistic verifiers that watch asymmetric debates (where one side is unable to see some of the messages of the other) between two deterministic provers, and try to determine who is right. We prove that probabilistic verifiers outperform their deterministic counterparts as asymmetric debate checkers. It is shown that the membership problem for every language in $\mathsf{NSPACE(s(n))}$ has a $2^{s(n)}$-time debate where one prover is completely blind to the other one, for polynomially bounded space constructible $s(n)$. When partial information is allowed to be seen by the handicapped prover, the class of languages debatable in $2^{s(n)}$ time contains $\mathsf{TIME(2^{s(n)})}$, so a probabilistic finite automaton can solve any decision problem in $\mathsf{P}$ with small error in polynomial time with the aid of such a debate. We also compare our systems with those with a single prover, and with competing-prover interactive proof systems. 
\end{abstract}

\section{Introduction}
\label{sec:intro}
The alternating computation model is well-known to correspond to the problem of finding winning strategies in two-person games of perfect information \cite{CKS81}. Another equivalent way \cite{Go08} of looking at this model is in terms of \textit{debate checking}, where a deterministic Turing machine (the "verifier" $V$ for $L$, the language under consideration) reads an exchange of arguments for and against the membership of the input string $w$ in $L$ written by two opposing "provers" (P1 and P0, respectively), and tries to determine who is right.\footnote{Note that $V$ is just a silent spectator who does not communicate any information to the provers; this is a crucial difference between the debate checking model and the multi-prover interactive protocols with competing provers such as those of \cite{FKS95}. See Section \ref{sec:model}.} While "stand-alone" polynomial-time Turing machines only recognize the languages in the complexity class $\mathsf{P}$, every language in $\mathsf{PSPACE}$ is "debatable" in the sense that P1  can make V accept no matter what P0 says if and only if $w \in L$. This advantage vanishes when the verifier is restricted to use a constant amount of space, where the corresponding class of debatable languages turns out \cite{LLS78} to equal those that can be recognized by stand-alone finite-state automata, namely, the regular languages.

The basic model described above has been extended in several ways. Reif \cite{Re79} studied two-person games of \textit{incomplete information}, where one of the players may not see all of the positions of its opponent, and \textit{zero information},\footnote{This is not to be confused with the unrelated notion of \textit{zero knowledge}.} which correspond to an extreme case where one of the players cannot see what the other one is doing at all. The problems of finding winning strategies in these new types of games correspond to generalizations of the alternating Turing machine (ATM) model that are called \textit{private} and \textit{blind} ATMs, respectively. Adopting the three-agent "verification" view mentioned above, naming what the verifier checks in these new setups requires a generalization of the notion of "debate." In the protocols to be presented in this paper, P0 can  see all the messages of P1, whereas P1 can only see the messages of P0 that are designated as public by P0. The verifier $V$ can of course see all messages of both provers. This may be likened to a debate where one of the debaters is deaf, and can "hear" his opponent only when the opponent chooses to write her message on a board for him to read. We call this process a (potentially) \textit{asymmetric debate}. In contrast with the situation with complete-information debates, the class of languages that are debatable in this generalized sense for constant-space verifiers has  been proven to properly contain those for stand-alone verifiers that are similarly bounded \cite{PR79}.

Another extension that has been considered for the original debate checking model is to upgrade the verifier to a probabilistic, rather than deterministic, TM. Condon et al. \cite{CFLS95,CFLS97} have shown that, when the verifier is a polynomial-time PTM that is restricted to use only logarithmically many random coins, and read only a constant number of bits of the debate, every language in $\mathsf{PSPACE}$ is debatable, even if P0 uses the simple strategy of choosing its messages randomly. 

In this paper, we examine the power of probabilistic verifiers that read potentially asymmetric debates between two deterministic provers. We focus on constant-space verifiers, since these already yield extremely powerful systems: When the debate format just requires that all and only the members of $L$ should be accepted with high probability, but allows P1 to talk endlessly when $w \notin L$, tricking $V$ to run forever without rejecting, all recursively enumerable languages have "one-way interactive proof systems" (i.e. debates without P0). 
We show that, for any space constructible and polynomially bounded function $s(n)$, all languages in the class $\mathsf{TIME(2^{s(n)})}$ have generalized debates that can be checked with constant-space verifiers in time $2^{s(n)}$, and $\mathsf{NSPACE(s(n))}$ has $2^{s(n)}$-time debates where P1 does not see P0. So, for instance, every language in $\mathsf{P}$ has a constant-space polynomial-time debate system, whereas this is not known to be possible even in exponential time for one-way interactive proof systems with similarly bounded verifiers. 

The rest of this paper is structured as follows: We provide the necessary definitions in Section \ref{sec:model}. Our results on the power of probabilistic verifiers of asymmetric debates are presented in Section \ref{sec:results}. In Section \ref{sec:conc}, we list some open questions.

\section{Preliminaries} 
\label{sec:model}

Our formal model of debate systems is a natural generalization of one-way interactive proof systems \cite{Co93B,CL89}. The verifier is a probabilistic Turing machine (PTM) with a read-only input tape and a single read/write work tape. The input tape holds the input string between two occurrences of an end-marker symbol, and we assume that the machine's transition function never attempts to move the input head beyond the end-markers. The input tape head is on the left end-marker at the start of the process. The verifier reads information written by two competing provers, P1 and P0, by alternately consulting two \textit{reading cells}, C1 and C0. When it is consulted for the $i$th time, cell C1 returns the $i$th symbol in P1's argument for the membership of the input string in the language under consideration. C0 similarly yields the next symbol in P0's counterargument. The symbols used in this communication are chosen from two mutually exclusive alphabets named $\Gamma$ and $\Delta$. $\Gamma$ is the set of symbols that can be seen by the opposing prover, and P1 is only allowed to use this alphabet in its messages. $\Delta$ is the set of private symbols that P0 may choose to write to the verifier without showing to P1 (this is the crucial component in our definition that allows for asymmetric debates). P0 can use any member of $\Gamma \cup \Delta $ in its messages. P1 (resp., P0) is assumed to have seen the subsequence of all the public symbols (i.e. those in $\Gamma$) among the first $i-1$ (resp., $i$) symbols emitted up to that time before preparing its $i$th symbol. Each prover can be modeled as a (not necessarily computable) prover transition function, which determines the $i$th symbol $\zeta$ based on the input string and the part of the communication that it has seen so far. Using this communication infrastructure, P1 attempts to convince the verifier that it is right, and P0 tries to prove P1 wrong. We allow the possibility that the cheating prover sends an infinite sequence of symbols, a contingency that could cause careless verifiers to run forever. The verifier also has access to a source of private random bits. The state set of the verifier TM is $Q$, containing, among others, two special halting states, $q_{a}$ (accept) and $q_{r}$ (reject). One of the non-halting states is designated as the start state. 
A subset $R$ of $Q$ that does not contain the halting states is designated to be the set of \textit{reading states}. $R$ is partitioned to two subsets, $R_1$ and $R_0$, where the program of the machine (to be described shortly) ensures that the first reading state to be entered is in $R_1$, and the sequence of reading states entered during computation alternates between members of $R_1$ and $R_0$. Whenever a state $q \in R_i$ is entered, the next symbol of P$i$'s message is written in reading cell C$i$. Let $\lozenge=\{-1,0,+1\}$ denote the set of possible head movement directions. The program of the verifier is formalized by the verifier transition function $\delta$ as follows: For $i \in \{0,1\}$ and $q \in R_i$, $\delta(q,\sigma,\theta,\zeta_i,b)=(q',\theta',d_{ih},d_{wh})$ means that the machine will switch to state $q'$, write $\theta'$ on the work tape, move the input head in direction $d_{ih} \in \lozenge$, and the work tape head in direction $d_{wh} \in \lozenge$, if it is originally in state $q$, scanning the symbols $\sigma$,  $\theta$, $\zeta_i$, in the input and work tapes, and the reading cell C$i$, respectively, and seeing the random bit $b$ as a result of the coin toss. If $q \in Q - (R \cup \{q_{a},q_{r}\})$, a restricted version of the verifier transition function described above that does not involve the symbol from the reading cell is applied. The probability that $V$ accepts an input string $w$ (i.e. ends up in $q_{a}$) as the result of watching a debate between P1 and P0 is denoted by $a^w_{(V,P1,P0)}$. $r^w_{(V,P1,P0)}$ denotes the probability that $V$ rejects $w$ in such a scenario.

We say that language $L$ \textit{has a partial-information debate with error probability} $\varepsilon$ if there exist a verifier $V$ and a prover P1 such that
\begin{enumerate}
	\item for every $w \in L$, and for any prover $\mathrm{P0^*}$, $a^w_{(V,P1,P0^*)} \geq 1-\varepsilon$,  and,
	\item for any prover $\Pi$, there exists a prover $\mathrm{P0_\Pi }$,  such that for every $w \notin L$, $r^w_{(V,\Pi,P0_\Pi)} \geq 1-\varepsilon$.
\end{enumerate}

\indent We will also see some systems where condition (2) above is relaxed, so that the verifier is allowed to fail to halt with high probability for $w \notin L$. 

$\mathsf{PDEB(s,t)}$ is the class of languages that have partial-information debates with error probability $\varepsilon$ for some $\varepsilon<\frac{1}{2}$, such that the verifier of the debate has a worst-case space bound of $O(s(n))$, and an expected time bound of $O(t(n))$. Classes of languages with \textit{zero-information debates} (i.e. where P0 never emits a symbol in $\Gamma$) have names of the form $\mathsf{ZDEB(s,t)}$.


A \textit{one-way interactive proof system} is simply a debate system without P0. (One may assume that consulting C0 always yields the same symbol in this case.)

We focus on verifiers with constant space. The following facts are known in this regard about one-way interactive proof systems:

\begin{fact}\label{fact:onewayIPupper}
All languages that have one-way interactive proof systems with constant-space verifiers are in
$\mathsf{NTIME(2^{2^{O(n)}})}$. \cite{CL89}
\end{fact}

\begin{fact}\label{fact:onewayIPlower}
Every language in $\mathsf{NSPACE(n)}$ has a one-way interactive proof system with a constant-space verifier.\footnote{This follows from a simplification of the proof of Theorem 3.12 in \cite{DS92}.}
\end{fact}

\begin{fact}
One-way interactive proof systems which are allowed to fail to halt for nonmembers of the debated language exist for all recursively enumerable languages. \cite{CL89}
\end{fact}

\begin{fact}
The class of languages with one-way interactive proof systems where the verifier is restricted to use a constant number of random bits is precisely $\mathsf{NL}$. \cite{SY12}
\end{fact}

(Two-prover) debate systems with deterministic verifiers have been studied extensively:

\begin{fact}\label{fact:detvercompletedeb}
When the verifier is a deterministic constant-space TM, and the debaters display complete information to each other, no nonregular language is debatable. \cite{LLS78}
\end{fact}

\begin{fact} \label{fact:detverzerdeb}
When the verifier is a deterministic constant-space TM, and P0 displays zero information to P1, the class of  debatable languages is $\mathsf{NSPACE(n)}$. \cite{PR79}
\end{fact}

\begin{fact}\label{fact:detverpardeb}
When the verifier is a deterministic constant-space TM, and P0 displays partial information to P1, the class of  debatable languages is $\mathsf{TIME(2^{O(n)})}$. \cite{PR79}
\end{fact}

Since complete-information debates with probabilistic constant-space verifiers include one-way interactive proof systems as a restricted case, it is clear that they outperform their counterparts with deterministic verifiers, by Facts \ref{fact:onewayIPlower} and \ref{fact:detvercompletedeb}. To our knowledge, this paper is the first examination of potentially asymmetric debate systems with probabilistic verifiers.

We will find it interesting to compare our debate checking model with the \textit{competing-prover interactive proof system} (CIPS) model of \cite{FS89}. The differences between this model and ours are that the verifier in a  CIPS engages in a two-way conversation with P1 and P0, who are not allowed to see each other's messages. Furthermore, one of the CIPS provers (the one who is telling the truth) is allowed to use random coins. A CIPS is allowed to fail to halt with high probability for nonmembers of the language under consideration. The following facts have been proven in \cite{FS89} regarding the CIPS model with constant-space verifiers:
\begin{fact}\label{fact:FSre}
Every recursively enumerable language has a CIPS. This remains true even if the verifier is deterministic.
\end{fact}
\begin{fact}\label{fact:FSrecursive}
Every recursive language has a CIPS with halting probability 1.
\end{fact}

\section{Results}
\label{sec:results}
We start by showing that constant-space probabilistic verifiers for zero-information debates are more powerful than their deterministic counterparts. Note the contrast between the following theorem and Fact \ref{fact:detverzerdeb}.

\begin{theorem}\label{theorem:ZeroInfTime} 
For any $k\geq 1$, and any space constructible bound $s(n)$ such that $s(n)=O(n^k)$,
\begin{equation*} 
\mathsf{NSPACE(s(n))} \subseteq \mathsf{ZDEB(1,2^{s(n)})}.
\end{equation*} 
\end{theorem}
\begin{proof}
Let $L$ be a  language recognized by an $s(n)$-space nondeterministic TM $M$ with a single work tape and separate read-only input tape. As in \cite{DS92} (Theorem 3.12, p. 817), we assume without loss of generality that $M$ keeps a counter in a separate track of its work tape, and that this counter is incremented after each action performed on the first track, even for configurations that are not reachable from the initial configuration. $M$ halts when this counter reaches $2^{cs(n)}$, where $c$ is a constant. Let $Q_M$ denote the state set of $M$.

We will describe a debate system with verifier $V$ for $L$. Prover P1 attempts to convince $V$ that the input string $w$ is in $L$ by repeatedly presenting what it claims to be an accepting computation path (ACP) of $M$ on $w$. An ACP is a sequence of configurations of $M$ on $w$ that starts with the initial configuration, ends with an accepting configuration, and has the property that every pair of adjacent configurations in it obeys $M$'s nondeterministic transition relation. Each configuration in the ACP is preceded by the symbol $\$$, which we assume is not a member of $M$'s work tape alphabet or state set. Configurations are strings of the form $u q_d x$, representing the work tape and current state of $M$, where $u$ and $x$ are strings over the work tape alphabet of $M$, and $q_d$ is a so-called \textit{state/direction symbol} whose position is used for representing the work tape head position. The value of $q_d$ in the $i$th configuration in an ACP indicates both the state of $M$ in that configuration, and the direction in $\lozenge$ traveled by the input head of $M$ during the transition from the ($i$-1)th configuration, for all $i>1$. The initial configuration consists of the single state/direction symbol corresponding to the start state of $M$, and head direction 0. The input head position is not represented, since $V$ will use its own input head to track this head. P1 is supposed to write  the separator symbol $\#$ between every two adjacent ACP in its messages.

P0's task is to show that something is wrong with the computation paths sent by P1. Since $V$ can check that a purported ACP "begins right" (i.e. with the initial configuration), and "ends right" (i.e. contains the accept state symbol in its last configuration) on its own, P0 attempts to convince $V$ that P1 is making a transition error between two configurations, or that P1 is trying to make $V$ run forever by reading an endless "configuration." In the following, these two types counterargument are called the \textit{transition error} and \textit{infinity} (or $\infty$) claims, respectively. Of course, $V$ should also be wary of the possibility that P0 may be trying to trick it to falsely reject a valid transition, or to cause it to enter a loop.

P1 and P0 talk symbol by symbol, and $V$ goes through the sequences of symbols that they emit in a parallel fashion. P0's alphabet ($\Delta$ in the definition in the previous section) is $\{0, \varsigma, \tau, \upsilon , \infty\}$. P0 writes 0 to match each symbol written by P1 until it comes to a point where it wishes to indicate an error in P1's sequence, as follows.

If P0 wishes to claim that P1 is making a transition error between two adjacent configurations, say, $\alpha$ and $\beta$, it writes $\varsigma$ in the position corresponding to the $\$$ preceding  $\alpha$, and then indicates the positions within these configurations whose examination would lead to the discovery of the error by the symbols $\tau$ and $\upsilon$, emitting 0's for the other positions.

For example, if P0 is claiming that the invalidity of the transition from $\alpha$ to $\beta$ can be detected by comparing the $j^{th}$, ($j$+1)$^{st}$ and ($j$+2)$^{nd}$ symbols of $\alpha$ with the $k^{th}$, ($k$+1)$^{st}$ and ($k$+2)$^{nd}$ symbols of $\beta$, the sequences produced by P1 and P0 will look as in Figure \ref{fig:messageExample}, where $\alpha_i$ and $\beta_i$ stand for the $i^{th}$ characters of configurations $\alpha$ and $\beta$, respectively. Of course, if P0 is honest, $j=k$, and dishonest P0's may well trick careless verifers to reject the input by giving these signals for carefully selected unequal $j$ and $k$ values.

\begin{figure}
\centering
\boxed{
\begin{array}{l}
P1 : ...\$\alpha_1\alpha_2...\alpha_j...\$\beta_1\beta_2...\beta_k...\$... \\
P0 : ...\varsigma\hspace{1 mm}0\hspace{1.7 mm}0\hspace{1.5 mm}...\hspace{0.5 mm}\tau\hspace{1 mm}...\hspace{0.5 mm}0\hspace{1 mm}0\hspace{1.5 mm}0\hspace{1 mm}...\hspace{0.5 mm}\upsilon\hspace{1 mm}...\hspace{0.3 mm}0...
\end{array}}
\caption{P0 claims a transition error between configurations $\alpha$ and $\beta$.}
\label{fig:messageExample}
\end{figure}

At this point, some readers may be worrying about how P0 "knows" that P1 will make a transition error at that particular pair of configurations even before P1 has written the first character of the first configuration in the pair. This is not a problem, since our definition of debate systems just requires such a P0 function to \textit{exist} for any particular P1 that makes an erroneous claim, and no P0 function should exist for an honest P1 giving a correct ACP. For any sequence of symbols that can be produced by a deterministic P1 containing a transition error, there is of course such a response sequence associated to a deterministic P0.

If P0 wishes to claim that P1 is trying to make $V$ run forever by giving an infinite-length "configuration" $\eta$,  it does not use the method above to indicate the transition error between the preceding configuration and $\eta$. Instead, P0 writes the symbol $\infty $ in response to the $\$$ preceding the infinite-length configuration. 

We now describe the verifier $V$. Let us call each segment of the debate delimited by P1 printing $\#$'s a \textit{round}. $V$ accepts immediately if P0 does not report an error of P1 (by either one of the two methods described above) in any single round,\footnote{Note that the last configuration presented in such a round must be accepting, since we assume that P1 must be honest if P0 is violating the protocol.} or if it notices similar easily detectable "syntax errors" by P0. $V$ rejects immediately if it sees P1 starting any round with an incorrect initial configuration, or if it detects P1 committing simple syntax errors such as printing a configuration without exactly one state/direction symbol.

$V$ performs a partial check of every transition between adjacent configuration pairs presented by P1 by just focusing on the three-symbol "window" that has a state/direction symbol in its middle in each configuration. If these two windows are consistent with the transition function of $M$ and the presently scanned input symbol, $V$ moves its input head (which is tracing the movement of the input head of $M$) in the specified direction on the tape, and continues execution. Otherwise, it rejects.

If P0 is reporting a transition error, $V$ tries to make sure that P0 is indeed pointing out to corresponding positions in the two successive configurations. 
$V$ compares these values by using a modified version of a technique by Freivalds \cite{Fr81}. Figure \ref{fig:messageExample} may be helpful during the following exposition, which refers to a number $l$ whose value will be fixed according to the desired error bound, as will be explained later.

When it sees that P0 has responded by $\varsigma$ to the configuration delimiter $\$$ from P1, $V$ starts flipping $4l$ coins on each of the first $j$ symbols of configuration $\alpha$. $V$ flips $l$ additional coins on the $j$th symbol corresponding to $\tau$. These $4lj+l$ coins will be called \emph{the first set of accept coins}. When scanning the relevant portion of the messages of P1, $V$ stores the symbols $\alpha_j$, $\alpha_{j+1}$, and $\alpha_{j+2}$ in its memory for use in a possible comparison. For the next configuration $\beta$, $V$ similarly flips $4lk+l$ coins that we name \emph{the second set of accept coins}. Parallelly to the tossing of the accept coins, $V$ flips another set of $2l(j+k)$ coins, which will be called \emph{the control coins}. 

After reading the $\upsilon$, $V$ decides whether to accept, control or continue:
\begin{itemize}
\item If one or both of the sets of accept coins contains all 0's, $V$ accepts. 
\item If both sets of accept coins contain at least one 1, and the set of control coins contains all 0's, $V$ proceeds to check if the two triples ($\alpha_j$, $\alpha_{j+1}$, $\alpha_{j+2}$) and ($\beta_k$, $\beta_{k+1}$, and $\beta_{k+2}$) provide evidence for an invalid transition of $M$. If the two triples are among those that can appear when two configurations in a legitimate run of $M$ are properly aligned, $V$ accepts. Otherwise, it rejects.
\item If the set of control coins, as well as both  sets of accept coins, all contain at least one 1, $V$ continues execution without comparing the triples, discarding all information it collected for this purpose.
\end{itemize} 

If $V$ arrives at the end of a round, it rejects if the final configuration is not accepting. Otherwise, it continues to the next round.

The only case where $V$ gives up simulating the input head of $M$ is when P0 writes $\infty$ in response to some $\$$ of P1, claiming that P1 will be emitting an infinite sequence of non-$\$$ symbols from this point on.
\begin{itemize}
\item If $s(n)=O(n)$,  $V$ inspects this claim by checking if the next $mn$ symbols emitted by P1 contains a $\$$ for some constant $m$ such that $mn > s(n)$ for all input lengths $n>0$. (The empty input can be handled separately.) This check can be performed by using $V$'s input tape as a ruler. $V$ accepts  if it sees a $\$$ in this scan. Otherwise, it rejects. 
\item If $s(n)$ is superlinear, $V$ enters a loop where each iteration involves reading $n$ symbols from P1, and flipping $rn$ coins. $V$ accepts if it sees any $\$$ during this process. $V$ rejects if all of the $rn$ outcomes for some iteration turn out to be 0.
\end{itemize}

This concludes the description of $V$. We now show that this constitutes a bounded-error debate system for $L$.
Let us establish first that $V$ halts with probability 1. When $w \in L$, in which case P1 will be repeating a valid ACP in an infinite loop, this is evident from the facts that $V$ terminates directly if P0 does not claim any error, the procedure associated with the P0 signal $\infty$ halts with probability 1, and each processing of the transition error claim involves a small but nonzero probability of halting. When $w \notin L$, P1 could hope to avoid the $\infty$ signal by not giving an infinitely long "configuration", but by making a single transition error to an infinite loop of $M$'s computation, thereby keeping the halting probability low. But this is not possible because of the way we defined $M$ in the beginning: Since we equipped $M$ with a robust counter that causes $M$ to halt when it maxes out, any infinite sequence of configurations purporting to describe a run of $M$ must contain not one, but infinitely many transition errors, which will of course be reported by P0, and whose combined processing will lead the halting probability to accumulate to 1.

We will now analyze the probabilistic procedures for the transition error and $\infty$ claims, and show that the probability that they lead to an incorrect decision is much lower than the probability of a correct decision. Consider Figure \ref{fig:messageExample} again. Let $A$ denote  the event of acceptance as the result of the processing of a single transition error claim. Let $T$ denote the event that the two triples from the two successive configurations are indeed tested for consistency during this procedure. According to the algorithm described above, 
\begin{equation}\label{eq:prA}
Pr[\text{A}] = 2^{-4lj-l} + (1- 2^{-4lj-l})2^{-4lk-l},
\end{equation}
and
\begin{equation}\label{eq:prT}
Pr[\text{T}] = (1-Pr[\text{A}]) 2^{-2l\left(j+k\right)}.
\end{equation}

If $w\in L$, P0 will make sure that $j\neq k$ during transition error claims. In this case,
\begin{align*}
Pr[\text{A$|$} j \neq k] > 2^{-l\left(4j+1\right)}+2^{-l\left(4k+1\right)-1} 
> 2^{-l\left(4j+1\right)-1}+2^{-l\left(4k+1\right)-1}.
\end{align*}
\noindent Since $j \neq k$, the value of one of $j$ and $k$ is at most $\frac{j+k-1}{2}$. Without loss of generality, we  assume $j$ to have  the smaller value, and we get $2^{-l\left(4j+1\right)-1} \geq 2^{-l\left(4\frac{j+k-1}{2}+1\right)-1} = 2^{-l\left(2j+2k-1\right)-1}$. Then, 
\begin{equation*}
Pr[\text{A$|$} j \neq k] > 2^{-l\left(2j+2k-1\right)-1}. 
\end{equation*}

The probability that $V$ will be tricked to test the incorrectly aligned triples indicated by P0, thereby erroneously rejecting the input, is
\begin{equation*}
Pr[\text{T$|$} j \neq k]  < 2^{-l\left(2j+2k\right)}.
\end{equation*}
Therefore, 
\begin{equation*}
Pr[\text{A$|$} j \neq k] > 2^{l-1} Pr[\text{T$|$} j \neq k],
\end{equation*}
\noindent that is, each incorrect claim about a transition error leads to an acceptance with probability at least $2^{l-1}$ times greater than the probability that it leads to a rejection.

If $w \notin L$, P1 should make transition errors to end its rounds with accepting configurations, and P0 raises the transition error claim with $j=k$. In this case, it is easy to see that $V$ accepts with probability 
\begin{equation*}
Pr[\text{A$|$ } j = k]  < 2^{-4lj-l+1},
\end{equation*}
and tests the transition among the two triples with probability
\begin{equation*}
Pr[\text{T $|$ } j = k]  > 2^{-4lj-1},
\end{equation*}
leading us to conclude
\begin{equation*}
Pr[\text{T$|$} j = k] > 2^{l-2} Pr[\text{A$|$ } j = k].
\end{equation*}
Therefore, each correct claim about a transition error leads to rejection with probability at least $2^{l-2}$ times greater than the probability of an incorrect acceptance. 

As for the infinity claim, the only way $V$ can reach an incorrect decision in this regard is the possibility of rejecting a proper ACP when $s(n)$ is superlinear. The control loop in the $\infty$-checking procedure described above can complete at most $\lfloor \frac{s(n)}{n} \rfloor$ iterations,\footnote{Since P1 is truthful in this case, a $\$$ will be encountered after at most $s(n)$ P1 symbols are scanned.} and if $V$ is unlucky enough to obtain "all 0's" for a set of $rn$ coin flips in any one of these iterations, it will reject incorrectly. The probability of this is
\begin{equation*}
\sum^{\lfloor \frac{s(n)}{n} \rfloor}_{i=1}{(1-2^{-rn})^{i-1} 2^{-rn}} < s(n)2^{-rn},
\end{equation*} 
which is exponentially small, since $s(n)$ is polynomially bounded. 

It is clear that any error bound $\epsilon < \frac{1}{2}$ can be achieved by tuning the constants $l$ and $r$ in these procedures.

Finally, we analyze the time complexity. Since $j$ and $k$ are at most $s(n)$ in Equations \ref{eq:prA} and \ref{eq:prT}, the probability of halting in any single processing of a transition error claim is $2^{-O(s(n))}$. There can be at most $2^{O(s(n))}$ configuration descriptions of length $O(s(n))$ between any two adjacent transition error claims. The $\infty$ check takes linear time for $s(n)=O(n)$, and $2^{O(rn)}$ time for superlinear $s(n)$. We conclude that $V$ halts in expected time $2^{O(s(n))}$ in all cases. 
\end{proof} 

Note that the proof of Theorem \ref{theorem:ZeroInfTime} depends crucially on the inability of P1 to see P0's messages. If P1 could see P0, it could detect when P0 claimed a transition error is about to start, and present a correct transition at that point, branding P0 a liar. It is also interesting to note that, since P1 is deterministic, P0 is in fact not required to see the output of P1 either. 

In the next theorem, we examine the power of asymmetric debates by allowing P0 to display part (but not all) of its messages to P1.

\begin{theorem}\label{theorem:partial}
For any $k\geq 1$, and any space constructible bound $s(n)$ such that $s(n)=O(n^k)$,
\begin{equation*} 
\mathsf{TIME(2^{s(n)})} \subseteq \mathsf{PDEB(1,2^{s(n)})}.
\end{equation*} 
\end{theorem}
\begin{proof}
We modify the construction of Theorem \ref{theorem:ZeroInfTime} by setting up a debate about the computation of an alternating, rather than nondeterministic, Turing machine $M$ with space bound $s(n)$ on the input $w$. The result then follows using the fact \cite{CKS81} that $\mathsf{TIME(2^{s(n)})}=\mathsf{ASPACE(s(n))}$ for any $s(n)=\Omega(\log n)$.

We assume without loss of generality that the start state of $M$ is existential, no two existential or universal moves follow each other, and there are exactly two outgoing transitions from all non-halting states.

As in \cite{DS92}, P1 will try to convince the verifier of the membership of $w$ in the language $L$ of $M$ by showing that it can always pick a  sequence of existential transitions of $M$ that end up in an accept state, even if the universal transitions are selected by an adversary. The protocol begins with P1 writing a sequence of symbols representing the first two configurations in an ACP of $M$, thereby making the existential choice in the first transition. P0 then announces its choice about the next (universal) move of $M$ to P1 by using one of two prespecified symbols from the public communication alphabet $\Gamma$. P1 replies with the third and fourth configurations in the ACP, and the conversation goes on in this manner.  All through this public exchange  between P0 and P1 about the path for $M$ to follow, P0 is privately emitting its claims about the errors that P0 is committing, using a suitably adapted version\footnote{Note that $V$ sometimes needs to check for consistency between triples that correspond to configurations that are not one, but two $M$-moves away from each other in this new protocol.} of the method in the proof of Theorem \ref{theorem:ZeroInfTime}. 

The runtime and error bound analyses of the previous proof need no modification in this proof.
\end{proof}
   
Contrasting Theorem \ref{theorem:partial} with Fact \ref{fact:detverpardeb}, we see that probabilistic verifiers outperform their deterministic counterparts on partial-information debates as well.

Our techniques above can be seen easily to have an implication regarding a variant of the CIPS model. The proof of Fact \ref{fact:FSrecursive} in \cite{FS89} makes crucial use of the feature of that model which forbids the provers from seeing each other's messages. We can show that a new model obtained by allowing the P0 of the CIPS model to see all of P1's messages would still have the capability of handling all recursive languages, even if we downgrade the truthful prover to be deterministic.
\begin{theorem}
Every recursive language has a competing-prover interactive proof system where the provers are deterministic, the verifier is a probabilistic finite automaton that halts with probability 1, P0 displays zero information to P1, and P1 displays complete information to P0.
\end{theorem}
\begin{proof}
Let $M$ be a deterministic Turing machine that halts for all inputs.
We will use the protocol in the proof of  Theorem \ref{theorem:ZeroInfTime}. P1 will present ACP's of $M$ on the input string $w$, P0 will point out transition errors in the messages of P1, and the verifier $V$ will essentially process the transition error claims all in the same manner as in that proof. The difference is that P0 will not be making infinity claims, since the methods for processing those in  Theorem \ref{theorem:ZeroInfTime} are only applicable if $M$ has a polynomial space bound. Instead, we use the following technique from \cite{FS89}.

After each move, $V$ tosses a coin, and uses its new right to talk to the provers to announce the outcome to both provers. If the outcome is 1, $V$, P0 and P1 continue with the protocol as usual. Otherwise, they restart the protocol, with P1 emitting the first symbol of a string that it claims to be an ACP. Suppose that a correct description of the (accepting or rejecting) computation path of $M$ for $w$ in the proper format for P1 is $t$ symbols long. Either P1 or P0 must then tell a lie within the first $t$ symbols of their messages. There will be infinitely many contiguous subsequences of $t$ 1's appearing one after the other within the sequence of the announced coin outcomes. In each of these runs of 1's, $V$ will halt with a small probability, but the decision that it gives with that small probability will be overwhelmingly likely to be correct, as we saw in the analysis of the processing of the transition error claims in the proof of  Theorem \ref{theorem:ZeroInfTime}. Overall, $V$ halts with probability 1, and its error bound can be tuned down to be any desired nonzero value.
\end{proof}

\section{Concluding Remarks} 
\label{sec:conc}
We initiated the study of probabilistic checking of asymmetric debates, where the messages of one of the competing provers may be partially or completely hidden from the other prover. We proved that constant-space probabilistic verifiers outperform deterministic ones on both zero-information and partial-information debates. Some links with the related model of competing-prover interactive proof systems were also examined.

The following is an incomplete list of open questions related to this work:

\begin{itemize}
\item We do not know if Theorems \ref{theorem:ZeroInfTime} and \ref{theorem:partial} can be extended for space bounds that are not
polynomially bounded.
\item Although we have established some lower bounds, full characterizations of the power (with and without the simultaneous time restrictions) of complete-, zero-, and partial-information debate systems are still open for investigation.
\item Our protocols depend on the fact that the provers are deterministic, i.e. that they are functions from the input and the previous communication to the next symbol to be emitted. When we allow P1 to use a coin to determine whether to make an error in the configuration to be printed next, we no longer have a guarantee that there exists a corresponding P0 that will lead $V$ to rejection with high probability. This is reminiscent of the fact that the CIPS models of \cite{KLSRS00} seem to be weakened if the provers are allowed to have access to secret sources of unbiased random bits, although note that even a public coin seems enough in our case. Determining the power of debates with probabilistic verifiers is an open question.
\end{itemize}

\bibliographystyle{plain}
\bibliography{YakaryilmazSay}

\begin{thebibliography}{10}

\bibitem{CKS81}
Ashok~K. Chandra, Dexter~C. Kozen, and Larry~J. Stockmeyer.
\newblock Alternation.
\newblock {\em Journal of the ACM}, 28(1):114--133, 1981.

\bibitem{Co93B}
Anne Condon.
\newblock The complexity of the max word problem and the power of one-way
  interactive proof systems.
\newblock {\em Computational Complexity}, 3(3):292--305, 1993.

\bibitem{CFLS95}
Anne Condon, Joan Feigenbaum, Carsten Lund, and Peter Shor.
\newblock Probabilistically checkable debate systems and nonapproximability of
  \mbox{PSPACE}-hard functions.
\newblock {\em Chicago Journal of Theoretical Computer Science}, (4), 1995.

\bibitem{CFLS97}
Anne Condon, Joan Feigenbaum, Carsten Lund, and Peter Shor.
\newblock Random debaters and the hardness of approximating stochastic
  functions.
\newblock {\em SIAM Journal on Computing}, 26(2):369--400, 1997.

\bibitem{CL89}
Anne Condon and Richard~J. Lipton.
\newblock On the complexity of space bounded interactive proofs.
\newblock In {\em Proceedings of the 30th Annual Symposium on Foundations of
  Computer Science}, SFCS '89, pages 462--467. IEEE Computer Society, 1989.

\bibitem{DS92}
Cynthia Dwork and Larry Stockmeyer.
\newblock Finite state verifiers $\mbox{I}$: The power of interaction.
\newblock {\em Journal of the ACM}, 39(4):800--828, 1992.

\bibitem{FS89}
Uriel Feige and Adi Shamir.
\newblock Multi-oracle interactive protocols with space bounded verifiers.
\newblock In {\em Structure in Complexity Theory Conference}, pages 158--164,
  1989.

\bibitem{FKS95}
Joan Feigenbaum, Daphne Koller, and Peter Shor.
\newblock A game-theoretic classification of interactive complexity classes
  (extended abstract).
\newblock In {\em Proceedings of the Tenth Annual IEEE Conference on
  Computational Complexity}, pages 227--237, 1995.

\bibitem{Fr81}
R\={u}si\c{n}\v{s} Freivalds.
\newblock Probabilistic two-way machines.
\newblock In {\em Proceedings of the International Symposium on Mathematical
  Foundations of Computer Science}, pages 33--45, 1981.

\bibitem{Go08}
Oded Goldreich.
\newblock {\em Computational Complexity: A Conceptual Perspective}.
\newblock Cambridge University Press, 2008.

\bibitem{KLSRS00}
Marcos Kiwi, Carsten Lund, Daniel Spielman, Alexander Russell, and Ravi
  Sundaram.
\newblock Alternation in interaction.
\newblock {\em Computational Complexity}, 9:202--246, 2000.

\bibitem{LLS78}
Richard~E. Ladner, Richard~J. Lipton, and Larry~J. Stockmeyer.
\newblock Alternating pushdown automata.
\newblock {\em Foundations of Computer Science, Annual IEEE Symposium on},
  0:92--106, 1978.

\bibitem{PR79}
Gary~L. Peterson and John~H. Reif.
\newblock Multiple-person alternation.
\newblock In {\em Proceedings of the 20th Annual Symposium on Foundations of
  Computer Science}, SFCS '79, pages 348--363. IEEE Computer Society, 1979.

\bibitem{Re79}
John~H. Reif.
\newblock Universal games of incomplete information.
\newblock In {\em Proceedings of the eleventh annual ACM symposium on theory of
  computing}, STOC '79, pages 288--308. ACM, 1979.

\bibitem{SY12}
A.~C.~Cem Say and Abuzer Yakary{\i}lmaz.
\newblock Finite state verifiers with constant randomness.
\newblock In {\em How the World Computes}, volume 7318 of {\em Lecture Notes in
  Computer Science}, pages 646--654. 2012.

\end{thebibliography}

\end{document}